\documentclass[letterpaper, 10 pt, conference]{ieeeconf}

\IEEEoverridecommandlockouts                              % This command is only needed if 

\usepackage{tikz}
\usepackage{textcomp}

\usetikzlibrary{matrix}
\usetikzlibrary{angles,quotes}

\usepackage{graphics,graphicx} % for pdf, bitmapped graphics files
\usepackage{epsfig} % for postscript graphics files
\usepackage{mathptmx} % assumes new font selection scheme installed
\usepackage{times} % assumes new font selection scheme installed
\usepackage{amsmath} % assumes amsmath package installed
\usepackage{amssymb,nth, mathtools,dsfont}  % assumes amsmath package installed
\usepackage{amsbsy}
\usepackage{epstopdf}
\usepackage[noadjust]{cite}
\usepackage{romannum}
\usepackage{xcolor}
\usepackage{subcaption}
\usepackage{multirow}

\newcommand{\bmtx}{\begin{bmatrix}}
\newcommand{\emtx}{\end{bmatrix}}
\newcommand{\bsmtx}{\left[ \begin{smallmatrix}} 
\newcommand{\esmtx}{\end{smallmatrix} \right]}
\newcommand{\field}[1]{\mathbb{#1}}
\newcommand{\R}{\field{R}}

\newcommand{\N}{\field{N}}

\newtheorem{theorem}{Theorem}

\newtheorem{lemma}{Lemma}

\title{Stability and Performance Analysis of Discrete-Time \\ ReLU Recurrent Neural Networks}

\author{Sahel Vahedi Noori, Bin Hu, Geir Dullerud, and Peter Seiler% <-this % stops a space
	\thanks{S. Vahedi Noori and P. Seiler are with the Department of Electrical Engineering \& Computer Science, at the University of Michigan ({\tt\small sahelvn@umich.edu;} and
		{\tt\small pseiler@umich.edu}). B. Hu is with the Department of Electrical and Computer Engineering at the University of Illinois at Urbana-Champaign ({\tt \small binhu7@illinois.edu}). G. Dullerud is with the
  the Department of Mechanical Science and Engineering at the University of Illinois at Urbana-Champaign ({\tt \small dullerud@illinois.edu})
  }
	}

\begin{document}

\maketitle

\begin{abstract}

This paper presents sufficient conditions for the stability and $\ell_2$-gain performance of recurrent neural networks (RNNs) with ReLU activation functions. These conditions are derived by combining Lyapunov/dissipativity theory with Quadratic Constraints (QCs) satisfied by repeated ReLUs. We write a general class of QCs for repeated RELUs using known properties for the scalar ReLU. Our stability and performance condition uses these QCs along with a ``lifted" representation for the ReLU RNN.
We show that the positive homogeneity property satisfied by a scalar ReLU does not expand the class of QCs for the repeated ReLU.  We present examples to demonstrate the stability / performance condition and study the effect of the lifting horizon.

\end{abstract}

%\begin{IEEEkeywords}
% Absolute stability; Lurye system; Kalman Conjecture; Cycle behaviour. 
%\end{IEEEkeywords}

%---------------------------------------------------------------------------------
\section{Introduction}

This paper considers the analysis of recurrent neural networks (RNNs) with ReLU activation functions.  These are modeled by the interconnection of a discrete-time, linear time-invariant (LTI) system in feedback with a repeated ReLU. The goal is to derive sufficient conditions to prove stability and performance (as measured by the induced $\ell_2$ gain) for this RNN.  This work is motivated by the increasing interest in RNNs for inner-loop feedback control. Such RNNs can potentially improve performance over more standard LTI controllers. However, the analysis of closed-loop stability and performance is challenging due to the nonlinear activation functions in the RNN. The sufficient conditions in this paper are one step to address this issue.

Our technical approach combines multiple ingredients in the existing literature.  First, we note that the scalar ReLU has slope restricted to $[0,1]$. Hence the repeated ReLU satisfies a known QC involving doubly hyperdominant matrices \cite{Willems:1968,Willems:71,kulkarni02}. Second, we present QCs that are specific to repeated ReLUs building on work in \cite{richardson23,ebihara21}. Next, we write the ReLU RNN using a ``lifted" representation over an $N$-step horizon \cite{khargonekar85}. This lifting allows us to construct more general QCs that hold for the ReLU across multiple time steps. Finally, we combine Lyapunov and dissipativity theory \cite{willems72a,willems72b,schaft99,khalil01} with these QCs to obtain sufficient conditions for stability and performance of the lifted ReLU RNN. Stability and performance of the original (unlifted) ReLU RNN follows from this analysis.

This paper adds to the broader literature on 
integral quadratic constraint (IQC) conditions \cite{megretski97,veenman16,scherer22,seiler15}.
Our stability/performance condition is similar to the discrete-time, IQC formulations in \cite{lessard2016analysis,Lee:20,taylor18,Carrasco:2016,fry17,hu17,fetzer17,carrasco19}.  The most closely related conditions are in \cite{Lee:20,vanscoy22}, both of which use lifting and static QCs (rather than dynamic IQCs). There is also growing literature on using these techniques to analyze NNs and RNNs \cite{soykens99,tan2024robust,yin2021stability,richardson23,ebihara21}. We specifically build on the results for ReLU RNNs in \cite{richardson23,ebihara21}. In particular, \cite{richardson23} derives QCs for ReLU and uses Lyapunov theory to prove stability for continuous-time ReLU RNNs.  We build on this work in discrete-time and also prove $\ell_2$ bounds. The work in \cite{ebihara21} also addresses discrete-time ReLU RNNs, similar to our paper, but uses small-gain arguments. In contrast, we use QCs and Lyapunov/dissipativity theory leading to related, but different conditions.

Our contributions to this existing literature are as follows.  First, we write a general class of QCs for repeated ReLU using known existing properties. We show that the positive homogeneity property satisfied by scalar ReLUs, i.e. $\phi(\beta v) = \beta v$ for all $\beta \ge 0$, does not expand the class of QCs for repeated ReLUs. Second, we use the discrete-time lifting, originally proposed in \cite{khargonekar85}, that maintains the state dimension. This is different from the liftings used previously in \cite{Lee:20,vanscoy22} for which the lifted state  grows with the lifting horizon.  Finally, we present examples to illustrate the effects of the ReLU QCs and lifting time horizon.
We show that the ReLU QCs can significantly reduce the conservatism as compared to QCs developed for the more general class of slope-restricted nonlinearities (although the performance is problem dependent).

\section{Notation}

% This material can be found in most standard texts on signals and systems, e.g. \cite{zhou96,dullerud99}.

This section briefly reviews basic notation regarding vectors, matrices, and signals.  Let $\R^n$ and $\R^{n\times m}$ denote the sets of real $n\times 1$
vectors and $n\times m$ matrices, respectively. Moreover, $\R_{\ge 0}^n$ and $\R_{\ge 0}^{n\times m}$ denote vectors and matrices of the given
dimensions with non-negative entries.   Let $D^{n}$ denote the set of $n\times n$ diagonal matrices and $D^n_{\ge 0}$ the subset
of diagonal matrices with non-negative entries.   Finally, $M\in \R^{n\times n}$ is a \emph{Metzler matrix} if the off-diagonal entries are non-negative, i.e. $M_{ij}\ge 0$ for $i\ne j$.  A matrix $M\in \R^{n\times n}$ is \emph{doubly hyperdominant} if the off-diagonal elements are non-positive, and both the row sums and column sums are non-negative. 

%The set of $n\times n$ symmetric, positive definite matrices is denoted by $S_+^n$. 

%Square $m$-dimensional matrices with non-positive off-diagonal elements are denoted by $Z^m$.
%with the diagonal subset $D_+^n \subset S_+^n$.  %For any $M\in \R^{n\times n}$ we define $He(M):=M+M^\top$.

Next, let $\N$ denote the set of non-negative integers.  Let $v:\N \to \R^n$ and $w:\N \to \R^n$ be real, vector-valued
sequences.  Define the inner product $\langle v,w \rangle : = \sum_{k=0}^\infty v(k)^\top w(k)$.  The
set $\ell_2$ is an inner product space with sequences $v$ that satisfy $\langle v,v\rangle < \infty$. The corresponding norm is
$\|v\|_2 := \sqrt{ \langle v,v \rangle}$. 

%Define $S(m,L)$, $0<m<L$, to be the set of continuously differentiable functions that are strongly convex with parameter m, and have Lipschitz gradients with parameter L.

%---------------------------------------------------------------------------------
\section{Problem statement}

Consider the interconnection shown on the left of
Figure~\ref{fig:LFTdiagram}  with a static nonlinearity $\Phi$
wrapped in feedback around the top channels
of a nominal system $G$.  This interconnection is
denoted as $F_U(G,\Phi)$.  The nominal part
$G$ is a discrete-time, linear time-invariant (LTI)
system described by the following state-space model:
\begin{align}
  \label{eq:LTInom}
  \begin{split}
    & x(k+1) = A\, x(k) + B_1\,w(k) +  B_2\, d(k) \\
    & v(k)=C_{1}\,x(k)+D_{11}\, w(k)+ D_{12} \,d(k)\\
    & e(k)=C_{2}\,x(k)+D_{21}\, w(k)+ D_{22}\,d(k),
  \end{split}
\end{align}
where $x \in \R^{n_x}$ is the state. The inputs are $w\in \R^{n_w}$
and $d\in \R^{n_d}$ while $v\in \R^{n_v}$ and $e\in \R^{n_e}$ are
outputs.  The static nonlinearity $\Phi:\R^{n_v} \to \R_{\ge 0}^{n_v}$ maps $v$ to $w$ elementwise by
$w_i=\phi(v_i)$ for $i=1,\ldots,n_v$ where $\phi:\R\to\R_{\ge 0}$ is the ReLU function. The ReLU, shown on the right of
Figure~\ref{fig:LFTdiagram}, is:
\begin{align}
\label{eq:ReLU}
\phi(v)= \left\{
  \begin{array}{ll}
    0 & \mbox{if } v < 0 \\
    v & \mbox{if } v \geq 0 
  \end{array} 
  \right. .
\end{align}
We refer to $\Phi$ as a repeated ReLU and $\phi$ as the scalar ReLU. The interconnection $F_U(G,\Phi)$ is known as a linear fractional transformation (LFT) in the robust control literature \cite{zhou96}. The interconnection has its roots in the Lurye decomposition used in the absolute stability problem \cite{khalil01}.

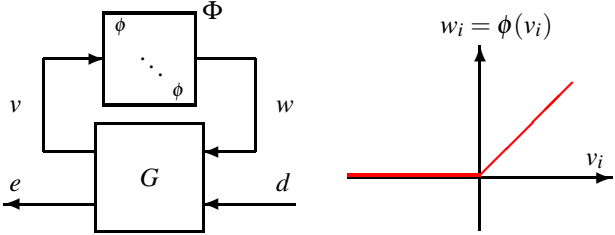
\begin{figure}[h!t]
\centering
\begin{picture}(230,90)(40,20)
 \thicklines
 \put(75,25){\framebox(40,40){$G$}}
 \put(143,40){$d$}
 \put(150,35){\vector(-1,0){35}}  
 \put(42,40){$e$}
 \put(75,35){\vector(-1,0){35}}  
% I/O for Repeated ReLU
% \put(80,75){\framebox(30,30){$\Phi$}}
    \put(78,73){\framebox(34,34){
    $\begin{smallmatrix}
    \phi & &  \\ & \,\, \ddots  & \\ 
    & & \phi 
    \end{smallmatrix}$
    }}
 \put(115,105){$\Phi$}
 \put(42,70){$v$}
 \put(55,55){\line(1,0){20}}  
 \put(55,55){\line(0,1){35}}  
 \put(55,90){\vector(1,0){23}}  
 \put(143,70){$w$}
 \put(135,90){\line(-1,0){23}}  
 \put(135,55){\line(0,1){35}}  
 \put(135,55){\vector(-1,0){20}}  
 % Add graph of ReLU
 \put(170,45){\vector(1,0){100}}  
 \put(260,50){$v_i$}
 \put(220,25){\vector(0,1){70}}  
 \put(205,100){$w_i=\phi(v_i)$}
 \put(170,46){{\color{red} \line(1,0){50}} }  
 \put(220,46){{\color{red} \line(1,1){35}} }  
\end{picture}
\caption{Left: Interconnection $F_U(G,\Phi)$ of a nominal discrete-time LTI system $G$ and repeated ReLU $\Phi$. Right: Graph of scalar ReLU $\phi$.}
\label{fig:LFTdiagram}
\end{figure}

This feedback interconnection involves an implicit equation if $D_{11}\ne 0$.  Specifically, the second equation in \eqref{eq:LTInom} combined with $w(k)=\Phi( v(k) )$ yields:
\begin{align}
   \label{eq:WellPosed}
   v(k)=C_{1}\,x(k)+D_{11}\, \Phi( v(k) )+ D_{12} \,d(k).
\end{align}
This equation is \emph{well-posed} if there exists a unique solution $v(k)$ for all  values of $x(k)$ and $d(k)$.  Well-posedness of this equation implies that the dynamic system $F_U(G,\Phi)$ is well-posed in the following sense:  for all
initial conditions $x(0)\in\R^{n_x}$ and  inputs $d\in \ell_2$ there exists unique solutions $x$, $v$, $w$ and $e\in \ell_2$ to the system $F_U(G,\Phi)$. There are simple sufficient conditions for well-posedness of \eqref{eq:WellPosed}, e.g. Lemma 1 in \cite{richardson23} (which relies on results in \cite{valmorbida18,zaccarian02}). Thus, we'll assume well-posed for simplicity in the remainder of the paper.

A well-posed interconnection $F_U(G,\Phi)$ is \emph{internally stable} if $x(k)\to 0$ from any initial condition $x(0)$ with $d(k)=0$ for $k\in \N$. In other words, $F_U(G,\Phi)$ is internally stable if $x=0$ is a globally asymptotically stable equilibrium point with no external input.  A well-posed interconnection $F_U(G,\Phi)$ has \emph{finite induced-$\ell_2$ gain} if there exists  $\gamma<\infty$ such that the output $e$ generated by any $d\in \ell_2$ with $x(0)=0$ satisfies $\|e\|_2 \le \gamma \, \|d\|_2$.  This bound
is denoted as $\|F_U(G,\Phi)\|_{2 \to 2} \le \gamma$.
The goal of this paper is to derive sufficient conditions, using the properties of the ReLU,
that verify $F_U(G,\Phi)$ is internally stable and has finite induced $\ell_2$ gain.
%---------------------------------------------------------------------------------
\section{Main Results}

The section presents the main results: sufficient conditions for stability and performance of the ReLU RNNs. These sufficient conditions are based on dissipativity theory and quadratic constraints (QCs) for the repeated ReLU.

\subsection{Quadratic Constraints for Slope-Restricted Functions}
\label{sec:QCslopeNL}

First, consider a scalar function $\phi:\R \to \R_{\ge 0}$  that
satisfies $\phi(0)=0$ and has slope restricted to $[0,1]$, i.e.:
\begin{align*}
  0 \le  \frac{ \phi(v_2)-\phi(v_1)}{v_2-v_1} \le 1, 
  \,\, \forall v_1, \, v_2 \in R, \, v_1\ne v_2.
\end{align*}
This subsection presents QCs that hold for any activation function $\phi$ that satisfies these conditions.  This includes, but is not limited to, the ReLU.

The graph of $\phi$ is restricted to the $[0,1]$ sector, i.e. its graph lies between lines that pass through the origin with slope 0 and 1. Thus the following QC holds $\forall v\in\R$ and $w=\phi(v)$:
\begin{align}
  \label{eq:sectorQC1}
     w (v-w) \ge 0.
\end{align}
Moreover, any two points on the graph of $\phi$ are connected by a line with slope in $[0,1]$. Hence the following QC holds for all $v_1,v_2\in\R$, $w_1=\phi(v_1)$, and $w_2=\phi(v_2)$:
\begin{align}
  \label{eq:slopeQC1}
    (w_2-w_1) \,
    \left( v_2-v_1-(w_2-w_1) \right) \ge 0.
\end{align}
%The sector and slope QCs in \eqref{eq:sectorQC1} and \eqref{eq:slopeQC1} hold even when mulitplied by any non-negative constant. 

Next, consider a repeated nonlinearity $\Phi:\R^m \to \R_{\ge 0}^m$ that maps $v$ to $w$ elementwise $w_k=\phi(v_k)$ for $k=1,\ldots,m$. We can derive a QC for $\Phi$ based on scaled combinations of the sector and slope constraints for the data $\{ (v_k, w_k) \}_{k=1}^m$. Specifically, for any non-negative
constants $\lambda_k$ and $\gamma_{kj}$:
\begin{align*}
&   \lambda_k \, w_k (v_k-w_k) \ge 0
    \,\mbox{ for } k=1,\ldots m \\
& \gamma_{kj}  \,  (w_k-w_j) \,
    \left( v_k-v_j-(w_k-w_j) \right) \ge 0.
    \,\mbox{ for } k,j=1,\ldots m.
\end{align*}
Summing over all these constraints gives an inequality
of the form $w^\top Q_0 (v-w) \ge 0$ where
\begin{align}
    (Q_0)_{kj} = \left\{
    \begin{array}{cc}
    \lambda_k + \sum_{k\ne j} (\gamma_{kj} + \gamma_{jk})
    & \mbox{if } k= j \\
    -(\gamma_{kj}+\gamma_{jk}) & \mbox{else}
    \end{array}
    \right. .
\end{align}
$Q_0$ is a symmetric, doubly hyperdominant matrix.  Lemma~\ref{lem:doublyhypQC}, stated next, provides a more general QC for this case in that it allows for $Q_0$ to be non-symmetric. 

\vspace{0.1in}
\begin{lemma}
\label{lem:doublyhypQC}
Let $\Phi:\R^{m} \to \R_{\ge 0}^m$ be a repeated nonlinearity with $\Phi(0)=0$ and slope restricted (elementwise) to $[0,1]$. If $Q_0 \in \R^{m \times m}$ is doubly hyperdominant then the following QC holds 
$\forall v\in\R^{m}$ and $w=\Phi(v)$:
\begin{align}
\label{eq:doublyhypQC}
\bmtx v \\ w \emtx^\top
\bmtx 0 & Q_0^\top \\ Q_0 & -(Q_0+Q_0^\top) \emtx
\bmtx v \\ w \emtx \ge 0.
%w^\top M (v-w) \ge 0.
\end{align}
\end{lemma}
\vspace{0.1in}
\begin{proof}
This follows from the results in Section 3.5 of \cite{Willems:71}. Specifically,  for any pair $v$ and $w=\Phi(v)$, define $\tilde{v} :=v-w \in \R^m$. Equation~\ref{eq:slopeQC1} implies that $\{ (\tilde{v}_k, w_k) \}_{k=1}^m$ is similarly ordered data, i.e. $\tilde{v}_k > \tilde{v}_j$ implies  $w_k \ge w_j$. Moreover, it follows from \eqref{eq:sectorQC1} that the data is unbiased, i.e.  $w_k \tilde{v}_k \ge 0$ for $k=1,\ldots,m$. Theorem 3.8 in~\cite{Willems:71} implies that if $Q_0$ is doubly hyperdominant then $w^\top Q_0 \tilde{v} \ge 0$.  Thus $w^\top Q_0(v-w) + (v-w)^\top Q_0^\top w \ge 0$.
\end{proof}
\vspace{0.1in}

The proof of Theorem 3.8 in~\cite{Willems:71} relies on a cyclic rearrangement inequality satisfied by similarly ordered, unbiased data. In fact, the doubly hyperdominance condition is the largest class of QCs that holds for slope-restricted functions whose graph passes through the origin.  A precise statement of this fact is given in Theorem 1 of \cite{kulkarni02}.\footnote{Theorem 1 in \cite{kulkarni02} is stated for repeated monotone nonlinearities. A similar fact holds for nonlinearities with slope restricted to $[0,1]$ by a transformation of the input-output data.}

\subsection{Quadratic Constraints for ReLU}

The previous subsection presented QCs for functions with
slope restricted to $[0,1]$. This section derives several additional QCs that are specific to the repeated ReLU. This builds on  prior work in \cite{richardson23,ebihara21}. The starting point for these QCs are the following properties for the scalar ReLU that have been used previously in the literature \cite{richardson23,drummond24,fazlyab2020safety,ebihara21}:

\begin{enumerate}
\item \emph{Positivity:} The scalar ReLU is non-negative for all inputs: $\phi(v)\ge 0$ $\forall v\in \R$.

\item \emph{Positive Complement:} The scalar ReLU satisfies $\phi(v)\ge v$ $\forall v\in \R$.

\item \emph{Complementarity:} In addition to the $[0,1]$ sector, the ReLU satisfies the stronger complementarity property: its graph is identically on the line of slope 0 (when $v\le 0$) or the line of slope 1 (when $v\ge 0$).  Thus, $\phi(v) \, (v-\phi(v))=0$ $\forall v\in\R$.

\item \emph{Positive Homogeneity:} The scalar ReLU is homogeneous for all non-negative constants: $\phi(\beta v)= \beta \phi(v)$ $\forall v\in \R$ and $\forall \beta \in \R_{\ge 0}$.
\end{enumerate}

Properties 1-3 can be combined to construct QCs for the repeated ReLU. Lemmas~\ref{lem:repReLUQC1} and \ref{lem:repReLUQC2} below are variations of
results in \cite{richardson23}.

\vspace{0.1in}
\begin{lemma}
\label{lem:repReLUQC1}
Let $\Phi:\R^{m} \to \R_{\ge 0}^m$ be a repeated ReLU. If $Q_1 \in D^m$ then the following QC holds  $\forall v\in\R^{m}$ and $w=\Phi(v)$:
\begin{align}
\label{eq:repReLUQC1}
\bmtx v \\ w\emtx^\top 
\bmtx 0 & Q_1 \\ Q_1 & -2Q_1 \emtx
\bmtx v \\ w\emtx = 0 .
\end{align}
\end{lemma}
\vspace{0.1in}
\begin{proof}
$Q_1$ is diagonal by assumption so that:
\begin{align*}
\bmtx v \\ w\emtx^\top 
\bmtx 0 & Q_1 \\ Q_1 & -2Q_1 \emtx
\bmtx v \\ w\emtx = 
    2\sum_{k=1}^m (Q_1)_{kk} w_k (v_k-w_k).
\end{align*}
This is equal to zero based on the complementarity property for each $\{ (v_k, w_k) \}_{k=1}^m$. Note that the diagonal entries of $Q_1$ are not necessarily non-negative.
\end{proof}
\vspace{0.1in}

\vspace{0.1in}
\begin{lemma}
\label{lem:repReLUQC2}
Let $\Phi:\R^{m} \to \R_{\ge 0}^m$ be a repeated ReLU. If
$Q_2=Q_2^\top$, $Q_3=Q_3^\top$, $Q_4\in \R_{\ge 0}^{m \times m}$ then the following QC holds 
$\forall v\in\R^{m}$ and $w=\Phi(v)$:
\begin{align}
\label{eq:repReLUQC2}
\bmtx v \\ w\emtx^\top 
\bmtx Q_2 & -(Q_2+Q_4^\top) \\ -(Q_2+Q_4) & 
Q_2+Q_3+Q_4+Q_4^\top \emtx
\bmtx v \\ w\emtx \ge 0 .
\end{align}
\end{lemma}
\vspace{0.1in}
\begin{proof}
The positivity and positive complement properties are
linear constraints on $\{ (v_k, w_k) \}_{k=1}^m$. These can be combined to form QCs on pairs of points. The  QCs below hold for any
non-negative
$(Q_2)_{kj}$, $(Q_3)_{kj}$, and $(Q_4)_{kj}$:
\begin{align}
  \label{eq:posQC1}
   &  (Q_2)_{kj} \, (w_k-v_k) (w_j-v_j) \ge 0, \\
  \label{eq:posQC2}
   &  (Q_3)_{kj} \, w_k w_j \ge 0, \\
  \label{eq:posQC3}
   &  (Q_4)_{kj} \, w_k (w_j-v_j) \ge 0.
\end{align}
Equations~\ref{eq:posQC1} and \ref{eq:posQC2} are equivalent to
$(w-v)^\top Q_2 (w-v) \ge 0$
and $w^\top Q_3 w \ge 0$, respectively. These QCs
are quadratic forms and hence we can assume $Q_2=Q_2^\top$
and $Q_3=Q_3^\top$ without loss of generality.  
Equation~\ref{eq:posQC3} is equivalent to
$w^\top Q_4 (w-v) + (w-v)^\top Q_4^\top w \ge 0$.
\end{proof}
\vspace{0.1in}

The next result combines all the QCs in the
Lemmas~\ref{lem:doublyhypQC}-\ref{lem:repReLUQC2}.  This is the most general QC based on the $[0,1]$ sector / slope restrictions combined with the
known properties for scalar ReLU. 

\vspace{0.1in}
\begin{lemma}
\label{lem:repReLUFin}
Let $\Phi:\R^{m} \to \R_{\ge 0}^m$ be a repeated ReLU. 
Let $Q_2=Q_2^\top$, $Q_3=Q_3^\top \in\R_{\ge 0}^{m \times m}$ and $\Tilde{Q}\in \R^{m\times m}$ be given 
with $\tilde{Q}$ a Metzler matrix. Then the following QC holds  $\forall v\in\R^{m}$ and $w=\Phi(v)$:
\begin{align}
\label{eq:repReLUFin}
\bmtx v \\ w\emtx^\top 
\bmtx Q_2 &  -\Tilde{Q}^\top-Q_2 \\ -\Tilde{Q}-Q_2 & Q_2+Q_3+\Tilde{Q}+\Tilde{Q}^\top \emtx
\bmtx v \\ w\emtx \ge 0.
\end{align}
\end{lemma}
\vspace{0.1in}
\begin{proof}
Gather all the QCs given in Lemmas~\ref{lem:doublyhypQC}-\ref{lem:repReLUQC2}. Define $\Tilde{Q}=Q_4-Q_0-Q_1$
where $Q_0 \in \R^{m\times m}$ is doubly hyperdominant,
$Q_1\in D^m$ and $Q_4 \in \R_{\ge 0}^{m\times m}$. The diagonal entries of $\tilde{Q}$ can be anything while the off-diagonal entries must be non-negative. Hence $\tilde{Q}$ is a Metzler matrix.  
\end{proof}
\vspace{0.1in}

Finally, the next lemma exploits the positive homogeneity property of the ReLU. Related results are given in Section 4.3 of \cite{ebihara21}.

\vspace{0.1in}
\begin{lemma}
\label{lem:repReLUHom}
Let $\Phi:\R^{m} \to \R_{\ge 0}^m$ be a repeated ReLU. Assume $\Phi$ satisfies the QC
defined by $M=M^\top \in \R^{2m \times 2m}$, i.e
$\forall v\in \R^m$ and $w=\Phi(v)$:
\begin{align}
\label{eq:repReLUHom1}
\bmtx v \\ w \emtx^\top 
\bmtx M_{11} & M_{12} \\ M_{12}^\top & M_{22} \emtx
\bmtx v \\ w \emtx \ge 0.
\end{align}
If $\Lambda\in D_{\ge 0}^m$ then $\Phi$ also satisfies the QC
defined by $\bsmtx \Lambda & 0 \\ 0 & \Lambda \esmtx M
\bsmtx \Lambda & 0 \\ 0 & \Lambda \esmtx \in \R^{2m\times 2m}$, i.e.
$\forall \bar{v}\in \R^m$ and $\bar{w}=\Phi(\bar{v})$:
\begin{align}
\label{eq:repReLUHom2}
\bmtx \bar{v} \\ \bar{w} \emtx^\top 
\bmtx \Lambda^\top M_{11} \Lambda & \Lambda^\top M_{12} \Lambda \\ 
(\Lambda^\top M_{12} \Lambda)^\top & \Lambda^\top M_{22} \Lambda \emtx
\bmtx \bar{v} \\ \bar{w} \emtx \ge 0    
\end{align} 
\end{lemma}
\vspace{0.1in}
\begin{proof}
Take any $\bar{v}\in\R^m$ and $\bar{w}=\Phi(\bar{v})$. 
Let $\Lambda$ is a diagonal matrix with non-negative entries.
Apply positive homogeneity of the scalar ReLU, element-wise,
to conclude that $\Lambda \bar{w} = \Phi( \Lambda \bar{v})$.
Next, define $v=\Lambda \bar{v}$ and $w=\Lambda \bar{w}$
so that $w=\Phi(v)$. The QC in \eqref{eq:repReLUHom1} holds, by assumption, for $(v,w)$. Substitute $v=\Lambda\bar{v}$ and $w=\Lambda\bar{w}$ to show the  QC in \eqref{eq:repReLUHom2} also holds for 
 $(\bar{v},\bar{w})$.
\end{proof}
\vspace{0.1in}

The positive homogeneity property seems to generalize the class of QCs for repeated ReLU by allowing for an additional scaling $\Lambda \in D_{\ge 0}^m$. However, this additional freedom does not increase the class of QCs if we use the results of Lemma~\ref{lem:repReLUFin}. Specifically, 
if $Q_2, Q_3 \in \R_{\ge 0}^{m \times m}$
then $\Lambda Q_2\Lambda, \Lambda Q_3\Lambda\in \R_{\ge 0}^{m \times m}$ for all $\Lambda \in D_{\ge 0}^m$. Similarly,  if $\tilde{Q}\in \R^{m\times m}$ is a Metzler matrix then $\Lambda \tilde{Q} \Lambda$ is also a Metzler matrix for all $\Lambda \in D_{\ge 0}^m$.  Thus positive homogeneity does not provide additional freedom when combined with the most general QC given in Lemma~\ref{lem:repReLUFin}, i.e. we get the same class of QCs.  However, the positive homogeneity scaling $\Lambda$ can provide an additional degrees of freedom when combined with more restricted subsets of QCs.  For example, suppose we use the restricted class of QCs based on slope restrictions (Lemma~\ref{lem:doublyhypQC}). In this case, the positive homogeneity property can provide a benefit because if $Q_0\in \R^{m\times m}$ is doubly hyperdominant and  $\Lambda \in D_{\ge 0}^m$ 
then $\Lambda Q_0 \Lambda$ is not necessarily doubly hyperdominant. In this case, the scaling introduced by 
positivity homogeneity (Lemma~\ref{lem:repReLUHom}) will generalize beyond the initial (restricted) class of slope-restricted QCs.

\subsection{Conditions for Stability and Performance}

This section presents a sufficient condition for the feedback connection $F_U(G,\Phi)$ to be stable and have bounded induced $\ell_2$-gain. The condition uses a lifted plant, defined below, combined with Lyapunov/dissipativity theory and the QCs derived in the previous subsection. 
%Similar results to prove stability with convergence rates have appeared in \cite{Lee:20,taylor18}. %(although they use a different lifting).

The first step is to create a lifted representation
for $F_U(G,\Phi)$ where $G$ is the LTI system \eqref{eq:LTInom}
and $\Phi$ is the repeated ReLU. Define the stacked vectors for both the input/output signals of $G$ as:
\begin{align*}
\begin{split}
    V_N(k) = & \bmtx v^\top(k) & v^\top(k-1) & \cdots & v^\top(k-N+1) \emtx^\top, \\
    W_N(k) = & \bmtx w^\top(k) & w^\top(k-1) & \cdots & w^\top(k-N+1) \emtx^\top, \\
    D_N(k) = & \bmtx d^\top(k) & d^\top(k-1) & \cdots & d^\top(k-N+1) \emtx^\top, \\
    E_N(k) = & \bmtx e^\top(k) & d^\top(k-1) & \cdots & d^\top(k-N+1) \emtx^\top.
\end{split}
\end{align*}
These vectors stack the signals over an $N$ step horizon from $k-N+1$ to $k$. 

The nominal plant dynamics can be lifted \cite{khargonekar85} to evolve the state $N$ steps from $x(k-N+1)$ to $x(k+1)$. The lifted nominal plant, denoted $G_N$, has the form:
\begin{align}
  \label{eq:LTILifted}
  \begin{split}
    & x(k+1) = A^N \, x(k-N+1) + B_{1,N}\,W_N(k) 
       +  B_{2,N} \, D_N(k) \\
    & V_N(k)=C_{1,N}\,x(k-N+1)+D_{11,N}\, W_N(k)
         + D_{12,N} \,D_N(k)\\
    & E_N(k)=C_{2,N}\,x(k-N+1)+D_{21,N}\, W_N(k)
          + D_{22,N}\,D_N(k).
  \end{split}
\end{align}
The state matrices ($B_{1,N}$, $B_{2,N}$, etc.)  can be constructed for the specified horizon $N$ from the state-matrices of the original nominal plant. The lifted plant given here has the same state dimension as the original plant but the input/output dimensions  are stacked, e.g. $W_N \in \R^{n_w N}$ and $B_{1,N} \in \R^{n_x\times n_wN}$. A related, but different, lifting is used in \cite{Lee:20,taylor18} where the state-dimension grows with the horizon $N$.

%\pjs{Give the expressions in an appendix if we have time.}  

The repeated ReLU can also be lifted. Define $\Phi_N:\R^{n_vN} \to \R^{n_v N}$ by applying the scalar ReLU elementwise to the input. In summary, the lifted system $F_U(G_N,\Phi_N)$ maps $D_N(k)$ to $E_N(k)$ based on the interconnection of
$G_N$ in \eqref{eq:LTILifted} and $W_N(k)=\Phi_N(V_N(k))$.

We next state the stability and performance condition for the lifted system.  
Define a linear matrix inequality (LMI) with
the lifted system of $G_N$:
\begin{align}
\label{eq:LiftedLMI}
\begin{split}
& LMI(P,M,\gamma^2) := \\
& \bmtx (A^N)^\top P (A^N)-P  & (A^N)^\top P B_{1,N} 
  &  (A^N)^\top P B_{2,N} \\ 
  B_{1,N}^\top P (A^N) & B_{1,N}^\top PB_{1,N} 
  & B_{1,N}^\top P B_{2,N} \\
  B_{2,N}^\top P (A^N) & B_{2,N}^\top P B_{1,N} 
& B_{2,N}^\top P B_{2,N}-\gamma^2 I\emtx \\
& + \bmtx C_{2,N}^\top \\ D_{21,N}^\top \\ D_{22,N}^\top\emtx
  \bmtx C_{2,N}^\top \\ D_{21,N}^\top \\ D_{22,N}^\top \emtx^\top
+  \bmtx C_{1,N}^\top & 0 \\ D_{11,N}^\top & I  \\ D_{12,N}^\top & 0 \emtx
M
  \bmtx C_{1,N}^\top & 0 \\ D_{11,N}^\top & I  \\ D_{12,N}^\top & 0 \emtx^\top
\end{split}
\end{align}
The next theorem provides an stability and performance condition for the ReLU RNN formulated with this LMI.  The proof uses a QC for the lifted ReLU and a standard Lyapunov / dissipation argument.

\vspace{0.1in}
\begin{theorem}    
    \label{thm:StabPerfConf}
Consider the ReLU RNN $F_U(G,\Phi)$ 
with the LTI system $G$ defined in \eqref{eq:LTInom}. Assume this interconnection is well-posed. Let $G_N$ 
and $\Phi_N : \R^{m} \to \R^{m}$ be the lifted system and lifted ReLU for some $N\in \N$ with dimension $m:=n_vN$.

Let $Q_2=Q_2^\top$, $Q_3=Q_3^\top \in\R_{\ge 0}^{m \times m}$ and $\Tilde{Q}\in \R^{m\times m}$ be given with $\Tilde{Q}$ a Metzler matrix.
Define $M=M^\top  \in \R^{2m \times 2m}$ as follows:
\begin{align}
\label{eq:MReLUQC}
M:=\bmtx Q_2 &  -\Tilde{Q}^\top-Q_2 \\ -\Tilde{Q}-Q_2 & Q_2+Q_3+\Tilde{Q}+\Tilde{Q}^\top \emtx,
\end{align}
Then the lifted ReLU $\Phi_N$ satisfies the QC  defined by $M$, i.e. 
$\bsmtx v \\ w\esmtx^\top M \bsmtx v \\ w\esmtx$
$\forall v\in\R^{m}$ and $w=\Phi_N(v)$.

Moreover, if there exists a positive semidefinite matrix $P=P^\top \in \R^{n_x\times n_x}$ and scalar $\gamma \ge 0$ such that $LMI(P,M,\gamma^2)<0$, then the ReLU RNN $F_U(G,\Phi)$ is internally stable and
has $\|F_U(G,\Phi)\|_{2\to 2} < \gamma$.
\end{theorem}    
\begin{proof}
This theorem is a standard dissipation result~\cite{schaft99,
willems72a, willems72b, khalil01,hu17}. A proof is given for completeness.

%The (2,2) block of the LMI is:
%\begin{align*}
%    B_{1,N}^\top P B_{1,N} + D_{21,N}^\top D_{21,N}
%    + \bmtx D_{11,N}^\top & I \emtx M 
%      \bmtx D_{11,N} \\I \emtx < 0.
%\end{align*}
%The first two terms are $\ge 0$.
%\pjs{What additional assumptions do we need on M here? If $M=\bsmtx I & 0 \\ 0 & -I \esmtx$ then we can use this to conclude that $D_{11,N}$ has norm less than 1 and hence the RELU equations have a unique solution by the contraction mapping theorem (Need to find the correct theorem to cite, probably NL textbook by Khalil or Sastry). Or we can just assume up front that $D_{11}$ has norm $<1$ so that well-posedness follows. This is easier (and fine for the analysis condition) but it makes application of this result for synthesis more difficult.  This is a minor point though and can probably be addressed later.} Hence the lifted interconnection is well-posed.

It follows directly from Lemma~\ref{lem:repReLUFin} that $\Phi_N$ satisfies the QC defined by $M$. We'll first use this QC and the LMI condition to show the lifted system $F_U(G_N,\Phi_N)$ is stable and with $\ell_2$ gain bounded by $\gamma$.
    
The LMI is strictly feasible by assumption and hence it
remains feasible under small perturbations: $LMI(P+\epsilon I,M,\gamma^2) + \epsilon I< 0$  for some sufficiently small $\epsilon >0$. Moreover, well-posedness of $F_U(G_N,\Phi_N)$ follows from the assumption that $F_U(G,\Phi)$ is well-posed. Hence 
the lifted system $F_U(G_N,\Phi_N)$
has a unique causal solution $(x, W_N, V_N, E_N)$ for any given initial condition $x_N(0)$ and input $D_N\in \ell_2^{n_dN}$.   

Define a storage function by $V\left(x\right) := x^\top(P+\epsilon I) x$.  Left and right
multiply the perturbed LMI by $[x(k)^\top, W_N(k)^\top, D_N(k)^\top]$ and its transpose.  The result, applying the lifted dynamics~\eqref{eq:LTILifted}, gives the following condition:
\begin{align*}
   & V\left(x(k+1)\right)  - V\left(x(k)\right) 
    -\gamma^2  D_N(k)^\top D_N(k)
   +  E_N(k)^\top E_N(k)  
    \\
& + \bmtx V_N(k) \\ W_N(k) \emtx^\top
    M \bmtx V_N(k) \\ W_N(k) \emtx 
    + \epsilon \bsmtx x(k) \\ W_N(k) \\D_N(k) \esmtx^\top \bsmtx x(k) \\ W_N(k) \\D_N(k) \esmtx
    \le 0
\end{align*}
The term involving $M$ is non-negative as the lifted ReLU satisfies the QC defined by $M$.  Thus the dissipation inequality simplifies to:
\begin{align}
\label{eq:DI}
\begin{split}
   & V\left(x(k+1)\right)-V\left(x(k)\right) +  E_N(k)^\top E_N(k)  \\
   & \le  (\gamma^2-\epsilon) D_N(k)^\top D_N(k)  - \epsilon x(k)^\top x(k) 
\end{split}  
\end{align}
Internal stability and the $\ell_2$ gain bound for the lifted system $F_U(G_N,\Phi_N)$ follow from this inequality.  Specifically, if $D_N(k)=0$ for all $k$ then \eqref{eq:DI} simplifies to the following Lyapunov inequality:
\begin{align*}
    V(x(k+1)) - V(x(k)) \le -\epsilon x(k)^\top x(k)
\end{align*}
Hence $V$ is a Lyapunov function and the lifted system is globally asymptotically stable (Theorem 27 in Section 5.9 of \cite{vidyasagar02}).

Next, assume $x(0)=0$ and $D_N \in \ell_2$. Summing \eqref{eq:DI} from $k=0$ to $k=T-1$ and using $V(x(0))=0$ yields:
\begin{align*}
   V\left(x(T)\right) +  \sum_{k=0}^{T-1} E_N(k)^\top E_N(k)  
 \le \sum_{k=0}^{T-1}  (\gamma^2-\epsilon) D_N(k)^\top D_N(k)  
\end{align*}
Note that $V(x(T))\ge 0$ because $P$ is positive semidefinite. Moreover, the right side is upper bounded by $(\gamma^2-\epsilon)\|D_N\|_2^2$ for all $T\in\N$.  This implies
that $E_N \in \ell_2$ and $\|E_N\|_2 < \gamma\|D_N\|_2$. 

In summary, the lifted system $F_U(G_N,\Phi_N)$ is internally stable and satisfies $\|F_U((G_N,\Phi_N)\|_{2\to 2} < \gamma$. The lifting is an  isomorphism and preserves signal norms: $\|d\|_2 = \|D_N\|_2$
and $\|e\|_2 = \|E_N\|_2$. It follows that the original system also satisfies the gain bound
$\|F_U((G,\Phi)\|_{2\to 2} < \gamma$. Moreover, 
the state of the lifted system corresponds to the evolving the state of the original system forward by $N$ steps.  Simple bounding arguments can be used to
show that if the lifted state converges to the origin then so does the original state. Hence internal stability of the lifted system also implies internal stability of the original system. 
\end{proof}
\vspace{0.05in}

\section{Numerical examples}

\subsection{Stability Analysis}

A variety of examples are given in  \cite{Carrasco:2016} to study the use of discrete-time Zames-Falb multipliers for analyzing stability of Lurye systems.  We will use Example 6 in Table 1 of \cite{Carrasco:2016} as a benchmark to illustrate our ReLU RNN stability and performance condition. 

Consider the Lurye system shown in Figure~\ref{fig:Lurye} where $\phi:\R\to\R$ is a nonlinear function, $\alpha$ is a non-negative scaling, and $G_{11}(z)=\frac{2z+0.92}{z^2-0.5z}$ is a discrete-time system.   The results in \cite{Carrasco:2016} focus on the class of nonlinearities that have slope restricted to $[0,1]$. The goal is to find the stability margin, i.e. the largest value of $\alpha\ge 0$ for which this Lurye system is stable for all nonlinearities in this class.  Their reported results are shown in
Table~\ref{tab:CarrascoStabResults}. The Circle and Tsypkin criteria both can be used to prove stability for $\alpha$ up to 0.6510. The Zames-Falb condition in \cite{Carrasco:2016} proves stability for $\alpha$ up to 1.087.\footnote{The results depend on the number of terms included in the Zames-Falb finite impulse response filter.  This is the best (largest) reported value.}  In fact, the Zames-Falb condition achieves the largest possible stability margin for the class of [0,1] slope restricted nonlinearities.  This follows because the Lurye system is unstable for $\phi(v)=v$ and $\alpha=1.087$.   The destabilizing value $\alpha=1.087$ is called the Nyquist gain.

\begin{figure}[h!t]
\centering
\begin{picture}(180,80)(0,-60)
 \thicklines
 \put(0,5){$d$}
 \put(0,0){\vector(1,0){32}}  
 \put(35,0){\circle{5}}
 \put(50,5){$v$}
 \put(38,0){\vector(1,0){42}}  
 \put(80,-15){\framebox(30,30){$\phi$}}
 \put(110,0){\vector(1,0){70}}  
  \put(145,5){$e=w$}
 \put(155,0){\line(0,-1){45}}  
 \put(155,-45){\vector(-1,0){20}}  
 \put(105,-60){\framebox(30,30){$\alpha$}}
 \put(105,-45){\vector(-1,0){20}}  
 \put(55,-60){\framebox(30,30){$G_{11}(z)$}}
 \put(55,-45){\line(-1,0){20}}  
 \put(35,-45){\vector(0,1){42}}  
 \put(23,-20){\line(1,0){8}}  
\end{picture}
\caption{Lurye System}
\label{fig:Lurye}
\end{figure}
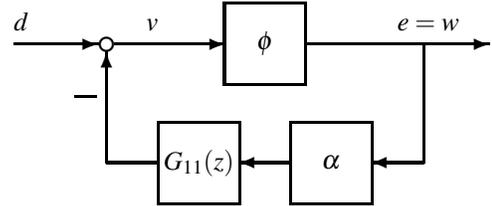

\begin{table}[h!]
\centering
\begin{tabular}{ |c|c c c c| } 
 \hline
Test type  & Circle & Tsypkin & Zames-Falb & Nyquist \\ \hline
  max $\alpha$  guarantee  & 0.6510 & 0.6510 & 1.0870 & 1.0870 \\ \hline
\end{tabular}
\caption{Stability margins:  Results from \cite{Carrasco:2016} for the Lurye system and nonlinearities with slope restricted to [0,1].}
    \label{tab:CarrascoStabResults}
\end{table}

Next, we used the condition in Theorem ~\ref{thm:StabPerfConf} to analyze the stability margin for the case when $\phi$ is the scalar ReLU. The Lurye system in Figure~\ref{fig:Lurye} is equivalent to the LFT representation $F_U(G,\phi)$ in Figure~\ref{fig:LFTdiagram} using:
\begin{align}
    G(z) = \bmtx -G_{11}(z) \, \alpha  & 1 \\ 1 & 0 \emtx
\end{align}
The Lurye system  with ReLU is stable for a given value of $\alpha$ if the $LMI(P,M,\gamma^2)<0$, defined in \eqref{eq:LiftedLMI}, is feasible for some $\gamma<\infty$. We solved for the largest feasible value of $\alpha$ using bisection using the ReLU QCs.\footnote{Note that $\gamma$ appears in the (3,3) block of the LMI in \eqref{eq:LiftedLMI}.  A Schur complement argument can be used to conclude that feasibility for some $\gamma<\infty$ is equivalent to feasibility of the upper left two blocks.  Hence our numerical implementation used bisection with only these upper left two blocks of the LMI.} 
The feasibility at each bisection step is a semidefinite program and was solved using CVX \cite{cvx14} as a front-end and SDPT3 \cite{toh99,tutuncu03} as the solver.  The bisection is initialized with $\underline{\alpha}=0$ and $\bar{\alpha}=200$. The bisection is terminated when $\bar{\alpha}-\underline{\alpha} \le 10^{-3} (1+\bar{\alpha})$. 

The top row  of Table~\ref{tab:LiftingStabResults} shows the results as a function of the lifting horizon.  For comparison, we also computed the stability margin using QCs defined by $M=\bsmtx 0 & Q_0^\top \\ Q_0 & -(Q_0+Q_0^\top) \esmtx$ where $Q_0$ is a doubly hyperdominant matrix. These results are shown in the second row of Table~\ref{tab:LiftingStabResults}.  Computational times are also given for each bisection to the stopping tolerance.

\begin{table}[h!]
\centering
\begin{tabular}{ |c|c c c c c| } 
 \hline
 lift size $N$ & 1 & 2 & 5 & 8 & 12\\ \hline
 ReLU: max $\alpha$  & 0.6516 & 0.6516 & 4.2999 & 33.472 & 181.543\\ 
 DH: max $\alpha$  & 0.6516 & 0.6516 & 0.8636 & 0.8820 & 0.8911\\ \hline \hline
 ReLU Comp. (s) & 4.73  & 4.87 & 5.15 & 5.10 & 5.32 \\
DH Comp. (s)  & 3.80 & 3.83 & 4.00 & 4.05 & 4.33 \\ \hline
 \end{tabular}
\caption{\emph{Stability margins (top):} Results using the doubly hyperdominant (DH) and ReLU QCs as a function of lifting horizon $N$. The ReLU QC exploits specific ReLU properties to significantly improve the stability margin. \emph{Computational times (bottom):}  given for each stability margin calculation.}
    \label{tab:LiftingStabResults}
\end{table}

The doubly hyperdominant QC is the strongest possible constraint for [0,1] slope restricted nonlinearities as discussed in Section~\ref{sec:QCslopeNL}.  This test corresponds to the circle criterion when $N=1$.  The stability margin improves with the lifting horizon, getting to a maximum value of 0.8911, but stays below the Nyquist value as expected.  The Zames-Falb condition, given in Table~\ref{tab:CarrascoStabResults}, provides  the larger stability margin of 1.0870 than our lifted Lyapunov condition with the doubly hyperdominance QC for $N=12$. 

More interestingly, the ReLU QC provides a stability margin that far exceeds the Nyquist value of 1.0870 as the lifting horizon $N$ increases, giving a maximum value of 181.543. This is possible because the ReLU QC exploits specific properties of the repeated ReLU that appears in the lifted system.  To empirically verify this result, we simulated the the Lurye system with $\phi$ as the scalar ReLU, $\alpha=100$, and 20 initial conditions for the state drawn from a $2\times 1$ Gaussian distribution with standard deviation of 10.  All simulations of the Lurye system with the ReLU converged back to the origin. This provides some independent validation of our stability margin results for the ReLU Lurye system.

\subsection{Gain Example}

In this section, we used the condition in Theorem ~\ref{thm:StabPerfConf} to analyze the induced $\ell_2$-norm for a ReLU RNN $F_U(G,\Phi)$.  The nominal part $G$ is a discrete-time, linear time-invariant (LTI)
system given in \eqref{eq:LTInom} with the following state matrices:
\begin{align*}
    A & = \bsmtx 0.84 & -0.17 & 0.10 & -0.04 \\ 0.17 & 0.80 & -0.11 & -0.04 \\ 0.05 & -0.11 & 0.90 & -0.08 \\ -0.04 & 0.18 & 0.01 & 0.74 \esmtx \\
    B_1 & = \bsmtx -0.08 & -0.18 \\ -0.11 & 0.11 \\ -0.01 & -0.05 \\ -0.04 & 0.03 \esmtx, \,\,
    B_2 = \bsmtx -0.04 & 0.11 & 0.14 \\ -0.06 & 0.03 & -0.02 \\ 0.02 & -0.03 & 0.01 \\ 0.04 & -0.08 & 0.20 \esmtx \\
    C_1 & = \bsmtx -0.35 & -0.40 & -1.04 & 1.36 \\ 0.90 & 0.55 & 1.13 & -0.46 \esmtx \\
    C_2 & = \bsmtx 0.79 & -0.62 & -1.87 & -0.09 \\ -0.67 & -1.05 & 1.80 & -0.06 \\ 2.25 & -1.08 & -0.36 & 1.08 \esmtx
\end{align*}
The nonlinearity $\Phi$ has $n_v=2$ inputs and $n_w=2$ outputs.  The dimensions of the disturbance and error channels are $n_d=3$ and $n_e=3$, respectively. The ``best" gain bound $\gamma$ for a specific lifting horizon $N$ is obtained by
minimizing $\gamma^2$ subject to the LMI constraint  $LMI(P,M,\gamma^2)<0$, as defined in \eqref{eq:LiftedLMI}. There are additional (linear) constraints on the QC matrices in $M$, e.g. $\tilde{Q}$ is a Metzler matrix.  This minimization is a semidefinite program and was solved using CVX \cite{cvx14} as a front-end and SDPT3 \cite{toh99,tutuncu03} as the solver. 
The top row  of Table~\ref{tab:LiftingGainResults} shows the results as a function of the lifting horizon.  For comparison, we also computed a bound on the $\ell_2$ gain using QCs defined by $M=\bsmtx 0 & Q_0^\top \\ Q_0 & -(Q_0+Q_0^\top) \esmtx$ where $Q_0$ is a doubly hyperdominant matrix. These results are shown in the second row of Table~\ref{tab:LiftingGainResults}.
The doubly hyperdominant QC is the strongest possible constraint for [0,1] slope restricted nonlinearities as discussed in Section~\ref{sec:QCslopeNL}. The results gain bound obtained using the ReLU QC is better (smaller) as it exploits specific properties of the ReLU.  Moreover, both results improve with the lifting horizon as the QCs can exploit couplings between the input/output data for the repeated nonlinearity.  Computational times are also given to run each minimization.  There is a mild growth in computation time with increasing horizon $N$ for this problem.

\begin{table}[h!]
\centering
\scalebox{0.93}{
\begin{tabular}{ |c|c c c c c c| } 
 \hline
 lift size $N$ & 1 & 2 & 3 & 4 & 5 & 6\\ \hline
 ReLU: gain bound & 7.556 & 5.530 & 3.932 & 3.466 & 3.247 & 3.128\\ 
 DH: gain bound & 34.379 & 13.450 & 7.992 & 5.844 & 4.811 & 4.263\\ \hline \hline
 ReLU Comp. (s) & 0.82  & 0.44 & 0.48 & 0.52 & 0.91 & 1.25 \\
DH Comp. (s)  & 0.28 & 0.28 & 0.29 & 0.36 & 0.47 & 0.61 \\ \hline
 \end{tabular}
 }
 \caption{\emph{Gain (top):} $\ell_2$-gain upper-bounds
 given using ReLU  and the doubly hyperdominant (DH) QCs as a function of lifting horizon $N$. The ReLU QC exploits specific ReLU properties to provide a significantly less conservative upper-bound on the gain. \emph{Computational times (bottom):}  provided for each upper-bound solution.  }
    \label{tab:LiftingGainResults}
\end{table}

\section{Conclusions}

This paper presents sufficient conditions for the stability and $\ell_2$-gain performance of RNNs with ReLU activation functions. These conditions are derived by combining Lyapunov/dissipativity theory with Quadratic Constraints (QCs) satisfied by repeated ReLUs. We use a ``lifted" representation for the ReLU RNN to derive our stability and performance condition. Future work will consider the computational cost of this condition and scalability to larger RNNs. We will also study the theoretical properties as the lifting horizon $N$ tends to infinity.

%\pjs{Comment that the lifting method that grows the state dimension will increase the computational cost but may also reduce the conservatism in the test. Add a few sentences on the Carrasco conjecture. There are also piecewise linear results for ReLU. What is the theoretical dependence on $N$?  Can we show that our lifting recovers ZF as $N\to \infty$.}

\section{Acknowledgments}

The authors acknowledge AFOSR Grant \#FA9550-23-1-0732 for funding of this work.

\bibliographystyle{IEEEtran}
\bibliography{references} 
 
\end{document}